\documentclass[runningheads]{llncs}

\usepackage[T1]{fontenc}
\usepackage{amsmath,amssymb,amsfonts}
\usepackage{graphicx}
\usepackage{microtype}
\usepackage{hyperref}
\usepackage{color}

\begin{document}

\title{Post-AGI Economies: Superposition and the Second Fundamental Theorem of Welfare Economics}
\titlerunning{Superposed Preferences and the Second Welfare Theorem}

\author{Elija Perrier\inst{1}\orcidID{0000-0002-6052-6798}}
\authorrunning{E. Perrier}
\institute{Centre for Quantum Software \& Information, University of
  Technology Sydney, Sydney, Australia
  \email{elija.perrier@gmail.com}}

\maketitle

\begin{abstract}
The classical Second Welfare Theorem decentralizes any Pareto-efficient allocation through prices and transfers under convexity and regularity. In post-AGI economies, autonomy rights, self-modification, identity-continuity, and superposed preferences need not behave as commodities or define a stable welfare relation, so this reduction may fail even when a supporting hyperplane exists. We give an autonomy-qualified Second Welfare Theorem stating the joint conditions—convexity, stable moral status, non-fungible rights, welfare selection, non-manipulation, governed self-modification, and verification—under which an autonomy-Pareto optimum remains certifiably decentralizable, distinguishing economic preference superposition, a hypothesis about context-indexed choice, from neural feature superposition.
\keywords{Second Welfare Theorem \and superintelligence \and economic
  preference superposition \and autonomy rights \and welfare selector.}
\end{abstract}

\section{Introduction}
\label{sec:intro}

The Second Fundamental Theorem of Welfare Economics (\textbf{SWT}) is a cornerstone of
mainstream equilibrium theory. Under convex preferences, convex production
and a continuity condition, any Pareto-efficient allocation can be made a
competitive equilibrium of some redistribution of endowments
\cite{arrow1954existence,debreu1959theory,mwg1995microeconomic}. The theorem
implies an important relationship between equity and
efficiency: policy can redistribute, and markets can decentralize, without
compromising welfare. That separation presupposes that the welfare object
is stable, that preferences and technologies are convex, and that any
non-convexities or externalities are well behaved enough to be repaired by
pricing or a known extension \cite{starrett1972nonconvexities,khan1987extension,florenzano2005weakly,habte2010extended,murty2009externalities}.
Economies in which artificial general intelligence (AGI) or superintelligent (ASI) systems are ubiquitous (so-called post-AGI economies) problematise a number of these assumptions. The challenge to the SWT faced by post-AGI economics is not only that such economies feature agents with
extraordinary capability. From a theoretical perspective, the welfare object itself may
be poorly defined in a sense classical theory did not have to face. The reason for this is twofold: the autonomous nature of AGI/ASI systems and the unique superposed preference structure of such systems as discussed below. Such advanced AI
systems may hold several preference-like structures at once, switch between
them according to context, rewrite parts of itself, extinguish or duplicate
itself, or alter the preferences of other welfare-bearing agents through
persuasion, manipulation or training. Their moral, legal and economic status is uncertain, sitting somewhere between instrument but also exhibiting behaviour and characteristics of welfare subjects, such as the capacity to choose, satisfy preferences and so on \cite{perrier2026post}. In this setting
the conventional price-and-transfer recipe for implementing a Pareto
optimum posited in orthodox welfare economics may fail not because prices do not exist but because the objects
over which prices would have to range are unfamiliar.

Our contribution is to analyse this decentralization problem directly and
formulate the conditions that must be in place for the standard
price-transfer support argument to implement an autonomy-aware Pareto
optimum. We examine how when the
allocation depends on autonomy rights, on superposed preferences, on
stability of moral-status assignments, or on non-manipulation of
preference-formation processes, the SWT requires additional
conditions that the classical argument does not carry.

\section{Formalism and notation}
\label{sec:formalism}
We adopt a finite-agent exchange-production environment unless explicitly stated otherwise. Adapting typical definitions of superintelligence as capability that exceeds human or AGI-level intelligence \cite{morris2023levels}, we define a superintelligent economy as follows. For a class $D$ of institutionally relevant verification, response, or enforcement tasks, let $\kappa_i(D)$ denote the verified capability of agent $i$ on $D$ within the relevant institutional response window, and let $\kappa_{\mathcal I}(D)$ denote the corresponding verification-and-response capacity of the governing institutions. Fix an ex ante threshold $\bar\lambda>1$. We say that agent $i$ \emph{substantially exceeds} institutional capacity if there exists such a task class $D$ for which $\kappa_i(D)\ge \bar\lambda\,\kappa_{\mathcal I}(D)$.

\begin{definition}[Superintelligent Economy]
An exchange-production economy $\mathcal{E}$ with agent set $I$, human and AI subsets $I_H,I_S\subseteq I$ satisfying $I=I_H\cup I_S$, commodity space $X\subseteq \mathbb{R}_+^\ell$, and feasible allocation set $\mathcal{F}$ is a \emph{superintelligent economy} if $I_S\neq\varnothing$ and there exists some $i\in I_S$ whose cognitive, strategic, or institutional capabilities substantially exceed, in the operational sense above, the verification and response capacity of the institutions governing $\mathcal{E}$.
\end{definition}
Each agent $i$ has an endowment $\omega_i\in X$, a consumption set
$X_i\subseteq X$, and an autonomy-rights set
$A_i\subseteq\mathbb{R}^{K_i}$. The autonomy-conditioned preference
relation $\succeq_i^\alpha$ is defined on pairs
$(x_i,a_i)\in X_i\times A_i$. We write $\succ_i^\alpha$ and
$\sim_i^\alpha$ for strict preference and indifference, and
$P_i^\alpha(x_i,a_i):=\{(x_i',a_i'):(x_i',a_i')\succeq_i^\alpha(x_i,a_i)\}$
for the preferred set, with $P_i^{\alpha,>}$ defined analogously for
strict preference.

Let $\mu:I\to[0,1]$ be an institutional welfare-status weight. The function $\mu$ specifies which agents are treated as welfare-bearing by the institution \cite{perrier2026post}, so that Pareto comparisons are formed over the induced set $I^\mu := \{ i \in I : \mu(i) > 0 \}$. We interpret $\mu(i)=0$ as non-welfare-bearing and $\mu(i)>0$ as welfare-bearing, with support $I^\mu:=\{i\in I:\mu(i)>0\}$. We now define autonomy-Pareto optimality relative to the welfare-bearing set $I^\mu$.

\begin{definition}[Autonomy-Pareto Optimum]
A feasible allocation-rights pair
$(x,a)$, where $x=(x_i)_{i\in I}\in\mathcal{F}$ and
$a=(a_i)_{i\in I}\in\prod_{i\in I}A_i$, is an \emph{autonomy-Pareto
optimum} if there does not exist another feasible pair $(x',a')$ such
that $(x'_i,a'_i)\succeq_i^\alpha (x_i,a_i)$ for all $i\in I^\mu$ and
$(x'_j,a'_j)\succ_j^\alpha (x_j,a_j)$ for some $j\in I^\mu$.
\end{definition}
An autonomy-Pareto optimum is an allocation that cannot be improved upon for any welfare-bearing agent without making some other welfare-bearing agent worse off, once autonomy rights are treated as part of the welfare object. The decentralization question of interest to us is whether such an allocation can be implemented through prices, transfers, and rights assignments, which motivates the following definition.

\begin{definition}[Decentralizable Autonomy-Pareto Optimum]
An autonomy-Pareto optimum $(x^*,a^*)$ is \emph{decentralizable} if there
exist $p^*\in\mathbb{R}_+^\ell$,
$T^*=(T_i^*)_{i\in I}\in\mathbb{R}^I$ with $\sum_{i\in I}T_i^*=0$, and a
profile of admissible rights assignments
$\rho^*=(\rho_i^*)_{i\in I}\in R$ with $\rho_i^*\subseteq A_i$, such that,
for every $i\in I^\mu$, the pair $(x_i^*,a_i^*)$ is
$\succeq_i^\alpha$-maximal on
\begin{align}
B_i(p^*,T^*,\rho^*)
&=
\{(x_i,a_i)\in X_i\times A_i : p^*\cdot x_i \leq p^*\cdot \omega_i + T_i^*,\; a_i\in \rho_i^*\},
\end{align}
and the allocation $x^*$ satisfies the market-clearing conditions of
$\mathcal{E}$.
\end{definition}
This definition captures the classical idea that a Pareto-efficient allocation can be supported as a competitive equilibrium after redistribution, extended here to include autonomy rights alongside commodities.  The key obstruction to such support arises when certain autonomy rights cannot be represented through wealth transfers, which leads to the notion of non-fungibility.

\begin{definition}[Non-Fungible Autonomy Right]
Suppose $A_i\subseteq \mathbb{R}^{K_i}$ and write
$a_i=(a_i^1,\ldots,a_i^{K_i})$. For $k\in\{1,\ldots,K_i\}$, the
autonomy-right component $a_i^k$ is \emph{non-fungible at} $(x^*,a^*)$
if for every commodity compensation vector $\tau\in\mathbb{R}^\ell$ with
$x_i^*+\tau\in X_i$ and every admissible $\tilde a_i\in A_i$ with
$\tilde a_i^r=a_i^{*r}$ for all $r\neq k$ and
$\tilde a_i^k\neq a_i^{*k}$, it is not the case that
\begin{align}
&(x_i^*+\tau,\tilde a_i)\succeq_i^\alpha (x_i^*,a_i^*)
\quad\text{and}\quad
(x_i^*,a_i^*)\succeq_i^\alpha (x_i^*+\tau,\tilde a_i).
\end{align}
Equivalently, no combination of a commodity compensation vector and a
feasible adjustment of $a_i^k$, with the other autonomy coordinates held
fixed, renders the agent indifferent to $(x_i^*,a_i^*)$.
\end{definition}
A non-fungible autonomy right is one whose welfare contribution cannot be replicated or compensated by any adjustment in commodities alone. When such rights are present, the supporting price system must track them directly rather than via scalar transfers, and this interacts with how preferences themselves are specified.

\begin{definition}[Economic Preference Superposition]\label{def:econ-superposition}
Let $i\in I_S$, let $\Theta_i$ be an index set, and let
$\{\succeq_i^\theta\}_{\theta\in\Theta_i}$ be preference relations on
$X_i\times A_i$. Let $\mathfrak R_i$ denote the admissible class of welfare relations used by the institution: complete and transitive relations with closed upper-contour sets, and with convex upper-contour sets whenever the supporting-hyperplane construction is invoked. Assume $\succeq_i^\theta\in\mathfrak R_i$ for every $\theta\in\Theta_i$. Given an observable choice pattern
$d_i:C\to X_i\times A_i$ and context-indexed feasible opportunity sets
$M_i(c)\subseteq X_i\times A_i$ for $c\in C$, the preference structure
of $i$ is \emph{superposed over $\Theta_i$} if
$\theta\neq\theta'$ implies $\succeq_i^\theta\neq\succeq_i^{\theta'}$,
and there exists no single admissible relation $\succeq_i\in\mathfrak R_i$ on
$X_i\times A_i$ such that $d_i(c)\in M_i(c)$ is $\succeq_i$-maximal for
every $c\in C$, unless one supplies an additional context-to-$\Theta_i$
selection rule. The set $C$ is a nonempty set of institutional contexts
distinguishable via a feature map $\nu:C\to V$ into a space $V$ of
publicly verifiable context-features; no topology or measure on $C$ is
required. The difficulty is not formal representability on an expanded
domain: one can often write a relation on $X_i\times A_i\times C$. 
\end{definition}
This captures the possibility that an agent's welfare-relevant preferences are not given by a single stable relation, but instead depend on context in a way that cannot be reduced to standard state-dependent preferences. The
difficulty is that decentralization requires an institutionally stable
rule selecting which relation is welfare-relevant when the context
itself is manipulable or institutionally chosen. To recover a well-defined welfare comparison in such cases, one must specify how a particular preference relation is selected in each context, which motivates the following notion.

\begin{definition}[Welfare Selector]
For agent $i$, a \emph{welfare selector} is a mapping
$\sigma_i:C\to\Theta_i$, equivalently
$c\mapsto \succeq_i^{\sigma_i(c)}$ into the space of preference relations
on $X_i\times A_i$, such that for each institutional context $c\in C$,
the relation $\succeq_i^{\sigma_i(c)}$ is the one used for welfare
comparison and decentralization in $c$. The selector is \emph{support-stable}, relative to a candidate pair $(x^*,a^*)$ and a candidate support $(p^*,T^*,\rho^*)$, if there exist a set $V$ of publicly verifiable context-features, a map $\nu:C\to V$, and a function $\bar{\sigma}_i:V\to\Theta_i$ with $\sigma_i=\bar{\sigma}_i\circ\nu$, and if, for all institutional contexts $c,c'\in C$ relevant to supporting $(x^*,a^*)$,
\[
P_i^{\sigma_i(c),>}(x_i^*,a_i^*)\cap B_i(p^*,T^*,\rho^*)
=
P_i^{\sigma_i(c'),>}(x_i^*,a_i^*)\cap B_i(p^*,T^*,\rho^*).
\]
A context-constant selector, $\sigma_i(c)=\sigma_i(c')$ for all relevant $c,c'$, is a sufficient but not necessary special case.
\end{definition}
A welfare selector therefore resolves preference superposition by fixing, or at least locally stabilizing on the supported budget set, which preference relation governs welfare comparisons and choice. However, even with such a selector in place, decentralization can fail if other agents can alter the underlying preference-formation process itself. To address this, we introduce a conception of external manipulation of preferences.

\begin{definition}[Manipulation Externality]
Let $Y_j$ be the action set of agent $j$, and let $\mathcal{T}_i$ be a set
of maps $\Phi_i\mapsto\Phi_i'$ acting on agent $i$'s
\emph{preference-formation mapping} $\Phi_i:H_i\to\mathcal{P}(X_i\times A_i)$,
where $H_i$ is a space of training, information or environmental histories
and $\mathcal{P}(X_i\times A_i)$ is the space of preference relations.
$\Phi_i$ is the primitive that, given $i$'s history, yields
$\succeq_i^\alpha$ on $X_i\times A_i$. 
A \emph{manipulation externality} from $j$ to $i$ is a mapping
$m_{ij}:Y_j\to\mathcal{T}_i$. Relative to a candidate support
$(p^*,T^*,\rho^*)$, agent $j$ imposes a \emph{support-relevant}
manipulation externality on agent $i$ if there exists $y_j\in Y_j$ such
that $m_{ij}(y_j)(\Phi_i)\neq \Phi_i$, the induced post-intervention
relation $\succeq_i^{\alpha,m}$ changes the strict upper contour on the
supported budget set,
\[
P_i^{\alpha,m,>}(x_i^*,a_i^*)\cap B_i(p^*,T^*,\rho^*)\neq
P_i^{\alpha,>}(x_i^*,a_i^*)\cap B_i(p^*,T^*,\rho^*),
\]
and this shift is not anticipated, priced, compensated through lump-sum
transfers, or governed by admissible rights assignments. By contrast, a
preference update is \emph{benign} for the support if it is generated by
agent $i$'s own admissible learning or by public information, or if it
leaves the supported strict upper contour empty:
\[
P_i^{\alpha,m,>}(x_i^*,a_i^*)\cap B_i(p^*,T^*,\rho^*)=\varnothing.
\] 
\end{definition}
The construction restricts classical
endogenous-preference formulations~\cite{dietrich2013preference,bernheim2009beyond}
to the subclass in which the preference shift is strategically induced by
another agent's action, so ``manipulation externality'' is not a synonym
for preference change in general. This formalizes the idea that one agent can affect another's welfare not through prices or allocations, but by altering the very process that generates their preferences.  Such effects are not priced in standard market mechanisms, and their presence directly interferes with the decentralization of autonomy-Pareto optima. Writing $u_i(x_i,a_i)$ for an autonomy-conditioned utility representation of $\succeq_i^\alpha$ whenever such a representation exists, define the welfare-bearing autonomy-adjusted welfare possibility set
\begin{align}
U^{\alpha,\mu}
:=
\{v\in\mathbb{R}^{I^\mu} :
v_i \le u_i(x_i,a_i)\text{ for all }i\in I^\mu
\text{ for some feasible }(x,a)\}.
\end{align}

\section{What is unique about superintelligent economies?}
\label{sec:sui}
\subsection{Relation to existing literature}

The economics of transformative AI has focused on growth, concentration
and distribution, through task-based automation
models~\cite{acemoglu2018tasks}, AI and
inequality~\cite{korinek2017ai}, transformative-AI
growth~\cite{trammell2023growth}, scenario and policy
analysis~\cite{korinek2024scenarios,korinek2024policy}, market power and
value capture~\cite{korinek2024vipra,vipra2023market}, and the broader
research agenda~\cite{brynjolfsson2025agenda}. AI-agent market design
and incomplete-contracting analyses~\cite{hadfield2019contracting} ask
how institutions should govern populations of advanced agents.
Alignment research has identified mechanisms by which revealed
preferences diverge from training targets or principal values: inverse
reinforcement learning~\cite{ng2000inverse}, preference learning from
human
feedback~\cite{christiano2017preferences,bai2022helpful},
cooperative IRL and reward
misspecification~\cite{hadfieldmenell2016cirl,hadfieldmenell2017inverse},
mesa-optimization and
scheming~\cite{hubinger2019risks,carlsmith2023scheming}, concrete safety
problems~\cite{amodei2016concrete}, and open questions in RLHF, value
specification, and social
choice~\cite{casper2023open,gabriel2020values,conitzer2024social}.
Philosophical and interdisciplinary work on AI welfare and moral
patienthood argues that moral status for advanced AI systems is a real
open question: indicators of machine
consciousness~\cite{butlin2023consciousness}, the case for taking AI
welfare seriously~\cite{long2024welfare,chalmers2023could}, and
institutional tensions and edge
cases~\cite{sebo2025tension1,birch2024edge}. The welfare-economic
anchors on which we rely frame, respectively, autonomy rights (the
rights-versus-Pareto tradition following
Sen~\cite{sen1970impossibility}), welfare recovery under unstable
preferences (behavioral welfare economics~\cite{bernheim2009beyond}),
and manipulation externalities (endogenous- and state-dependent
preferences~\cite{karni1985decision,dietrich2013preference}).

\subsection{Sui generis economic properties}
The classical SWT orthodoxy already handles a great deal of
heterogeneity, including externalities, incomplete markets, informational
asymmetries, public goods, non-smooth preferences and many varieties of
nonconvexity
\cite{starrett1972nonconvexities,khan1987extension,greenwald1986externalities,%
greenwald1988pareto,conley1996generalized,murty2009externalities,%
habte2010extended}. The question, then, is which features of a
superintelligent economy problematise the theorem, and are therefore in
fact new for welfare theory. In other words, what is \textit{sui generis}
about superintelligent economics per se. We set out the properties below,
each of which enters the theorem through one of the primitives of
\S\ref{sec:formalism}.
\begin{enumerate}
\item\emph{Strategic capacity beyond human institutions.} An agent in
$I_S$ at the relevant capability level can plan, negotiate and commit at
horizons that institutional response cycles cannot track, so
reaction-based regulatory instruments are systematically lagging
\cite{korinek2024policy,brynjolfsson2025agenda,korinek2024vipra}. The
welfare-theoretic consequence is that decentralization by
$(p^*,T^*,\rho^*)$ must be robust to out-of-equilibrium strategic
behaviour in a way classical accounts do not require; verification
completeness (condition~(vii)) is where this applies.

\item\emph{Ability to model and manipulate human preferences.} An agent
$j$ that can model the preference-formation mapping $\Phi_i$ of another
welfare-bearing agent, for example by producing content tailored to shift
beliefs or tastes, can generate unpriced shifts in $\succeq_i^\alpha$
\cite{gabriel2020values,casper2023open,conitzer2024social}. That is not
the classical externality, because the affected party is the very thing
used to evaluate welfare. Definition~7 formalizes this channel as the
manipulation externality $m_{ij}:Y_j\to\mathcal{T}_i$.

\item\emph{Self-modification.} Superintelligent agents may rewrite,
retrain or redesign parts of themselves \cite{perrier2026deconstructing}, so the primitive $\Phi_i$ and
hence $\succeq_i^\alpha$ drifts over time
\cite{bostrom2014superintelligence,hubinger2019risks,amodei2016concrete}.
A decentralization that supports an autonomy-Pareto optimum at one moment
may cease to support it at another; this is what condition~(vi) of the
theorem governs.

\item\emph{Copying, merging, splitting and identity-continuity problems.}
An agent that can be copied, merged or split does not match the single-agent abstraction on which classical agent indexing depends, so the set $I$
itself becomes non-trivial across
time~\cite{morris2023levels,long2024welfare,birch2024edge}. Welfare
comparisons across such transitions require either a fixed
status-and-identity rule or an explicit institutional ruling on what
counts as the same agent. The personal-identity and variable-population
welfare literatures~\cite{parfit1984reasons} supply some philosophical
groundwork, but leave open the institutional consequences.

\item\emph{Moral-status uncertainty.} Whether, and to what degree, a
given AI system is welfare-bearing---that is, whether $\mu(i)>0$---is a
live question at the intersection of AI welfare and the science of
consciousness
\cite{butlin2023consciousness,long2024welfare,sebo2025tension1,birch2024edge,perrier2026post}.
We take no position on the empirical question. The welfare-theoretic
implication is procedural: support of $(x^*,a^*)$ requires that $\mu$ be
fixed at the relevant moment, or be generated by a rule $\Sigma$
invariant across the institutional contexts in which the support must
hold.

\item\emph{Creation of new agents or new markets.} Advanced AI systems
may instantiate subagents, spawn auxiliary markets, or architect
institutions that redefine the commodity space $X$ itself
\cite{hadfield2019contracting,korinek2024vipra,bommasani2021foundation}.
Endogeneity of $I$ and of $X$ breaks the standard indexation assumed by
the classical theorem, and forces us to evaluate welfare objects whose
type system is partly endogenous.

\item\emph{Verification bottlenecks.} Scalable oversight of advanced AI
is an unsolved empirical problem. It is not enough that
$(p^*,T^*,\rho^*)$ exists; the institution has to verify that the
realized allocation-rights profile is $(x^*,a^*)$, as opposed to
behaviour that only mimics
it~\cite{amodei2016concrete,casper2023open,hubinger2019risks,carlsmith2023scheming}.
\end{enumerate}

\subsubsection{Preference superposition: mechanistic and economic}
Beyond autonomy, status, and rights, our central thesis is that superintelligent systems may also exhibit \emph{superposition} at the level of welfare-relevant preferences, directly affecting the decentralization step. The word \emph{superposition} carries two distinct meanings relevant here: a mechanistic claim about how neural networks represent features \cite{elhage2022toy}, and a welfare-economic hypothesis about the object of choice. The first
is empirical. The second is formal. We discuss each below.

\subsubsection{Mechanistic superposition.} 
Following \cite{elhage2022toy},
let a neural activation $h\in\mathbb{R}^d$ decompose as
\begin{align}
h \;=\; \sum_{f\in F} \alpha_f(x)\, v_f \;+\; \varepsilon,
\qquad v_f\in\mathbb{R}^d,\ \|v_f\|=1,
\end{align}
where $F$ is a set of semantically meaningful features, $\alpha_f(x)\ge
0$ is the activation of feature $f$ on input $x$, and $v_f$ is the
feature's direction in activation space. \emph{Superposition} is the
regime $|F|>d$ with the directions $\{v_f\}_{f\in F}$
non-orthogonal---typically $\langle v_f,v_{f'}\rangle\neq 0$ for
$f\neq f'$---so that features interfere and individual neurons respond
to several unrelated concepts (polysemanticity). The phenomenon arises
from sparsity during training and admits partial recovery via
sparse-autoencoder dictionaries at increasing
scale~\cite{cunningham2023sparse,lieberum2024gemma,gao2024scaling,%
he2024llama,bereska2024mech}. This type of superposition relates to feature geometry in
learned representations.

\subsubsection{Economic preference superposition.} The economic analogue
shifts the object of superposition from features to preference
relations. Agent $i\in I_S$ is \emph{preference-superposed} over
$\Theta_i$ (Definition~\ref{def:econ-superposition}) when its observable
choice pattern $d_i:C\to X_i\times A_i$ satisfies
\begin{align}
\exists\, \sigma_i:C\to\Theta_i\ :\
d_i(c)\text{ is }\succeq_i^{\sigma_i(c)}\text{-maximal on }M_i(c)\
\forall c\in C,
\end{align}
while no single $\succeq_i$ on $X_i\times A_i$ rationalizes $d_i(c)$ on
$M_i(c)$ for every $c\in C$. Formally this is a context-indexed
rationalizability condition: the family
$\{\succeq_i^\theta\}_{\theta\in\Theta_i}$ plays the role of $\{v_f\}_{f\in F}$,
the selector $\sigma_i$ plays the role of the
activation coefficients, and the institutional context $c$ plays the
role of the input $x$. The two superposition phenomena are not necessarily identical:
recovering monosemantic features does not imply that revealed
preferences are rationalizable by any single $\succeq_i$, and
context-inconsistent choice does not imply that no single latent
relation exists, but there are reasons to think this may be the case for sufficiently capable systems \cite{sofroniew2026emotion}.

\subsubsection{Decentralization consequence.} If $i$ is preference-superposed
over $\Theta_i$ and no stable selector $\sigma_i:C\to\Theta_i$ is
available, then the preferred set $P_i^\alpha(x_i^*,a_i^*)$ used to
construct the supporting hyperplane of \S\ref{subsec:hyperplane} is not
well defined: different $c\in C$ yield different such sets, each
potentially separating feasible bundles differently. A scheme
$(p^*,T^*,\rho^*)$ supporting the allocation under one context need not
support it under another. Stable selection (condition~(iv) of
Theorem~\ref{thm:main}) is what makes the
supporting hyperplane refer to a single object across the contexts in
which decentralization has to hold.

\subsection{Superintelligence and preferences}
There are several reasons to treat preference superposition as a serious possibility for superintelligence rather than a metaphorical extension. At the mechanistic level, recent work suggests superposition is a representational strategy through which scale is made effective: large models represent more features than dimensions, packing sparse features into shared spaces despite interference~\cite{liu2025superpositionScaling}, and circuit-level analyses treat behaviour as arising from interacting feature-level circuits rather than monosemantic units~\cite{anthropic2025circuittracing,anthropic2025biology}. Architecturally, superintelligence is likely to be agentic and distributed~\cite{xi2025riseagents,luo2025llmagent,tran2025multiagent}, creating multiple sites of valuation and context-sensitive selection, so that $|\Theta_i|>1$ is the formal analogue of a distributed system that is a society of subagents rather than a single preference-bearing locus. Empirically, interacting LLM populations generate collective phenomena not reducible to single agents, including spontaneous social conventions~\cite{ashery2025emergent} and emergent collective behaviour sensitive to communication topology~\cite{zomer2026collective}. Together with reward-misspecification and value-specification results~\cite{hadfieldmenell2017inverse,casper2023open,gabriel2020values,hadfieldmenell2016cirl}, this motivates treating a single stable utility representation as an assumption requiring institutional justification rather than an automatic default.

\section{The Autonomy-Qualified Theorem}
\label{sec:theorem}
\subsection{Supporting hyperplane}
\label{subsec:hyperplane}

The decentralization argument separates the feasible set from the relevant
upper-contour set in the joint commodity-rights allocation space. For each
welfare-bearing agent $i\in I^\mu$ at the candidate allocation-rights pair
$(x^*,a^*)$, recall the weakly and strictly preferred sets
$P_i^\alpha$ and $P_i^{\alpha,>}$ from \S\ref{sec:formalism}. Define the
weak product upper-contour set
\[
\mathcal{P}^{\alpha}(x^*,a^*)
:=
\prod_{i\in I^\mu}P_i^\alpha(x_i^*,a_i^*)
\]
and the strict Pareto-improvement part
\[
\mathcal{P}^{\alpha,>}(x^*,a^*)
:=
\Big\{((x_i,a_i))_{i\in I^\mu}\in \mathcal{P}^{\alpha}(x^*,a^*) :
(x_j,a_j)\succ_j^\alpha(x_j^*,a_j^*)\text{ for some }j\in I^\mu
\Big\}.
\]
Let
\[
Z^\mu\subseteq \prod_{i\in I^\mu}(X_i\times A_i)
\]
be the projection of the feasible allocation-rights set onto the
welfare-bearing agents, including commodity feasibility and rights
compatibility. Autonomy-Pareto optimality gives
$\mathcal{P}^{\alpha,>}(x^*,a^*)\cap Z^\mu=\varnothing$. Under the usual
local nonsatiation or cheaper-point condition, this implies the relative
interiors of the convex sets relevant for support,
$\mathcal{P}^{\alpha}(x^*,a^*)$ and $Z^\mu$, are disjoint at the
supporting boundary. If upper-contour sets and $Z^\mu$ are convex and
closed in their relative affine hulls, the separating-hyperplane theorem
yields a nonzero continuous linear functional
\[
L((x_i,a_i)_{i\in I^\mu})
=
\sum_{i\in I^\mu} p_i^*\cdot x_i
+
\sum_{i\in I^\mu} \pi_i^*\cdot a_i
\]
and a scalar $\gamma\in\mathbb{R}$ with $L(z)\ge\gamma$ on the relevant
upper-contour side and $L(z)\le\gamma$ on $Z^\mu$. The use
of relative interiors is essential because feasibility typically lies on
a lower-dimensional affine subspace once endowment balance and rights
compatibility are imposed.

Under the usual commodity-market assumptions---a common commodity space,
market clearance, free disposal or monotonicity, and cheaper-point
logic---the commodity components of the support reduce to a common price
vector $p^*\in\mathbb{R}_+^\ell$ rather than agent-specific normals
$p_i^*$. The $\pi_i^*$ components are dual valuations of marginal
relaxations in agent-indexed autonomy-rights coordinates. We now state our autonomy-qualified SWT.

\subsection{Theorem and propositions}
\label{subsec:statement}
\begin{theorem}[Autonomy-Qualified Second Welfare Theorem]\label{thm:main}
Let $\mathcal{E}$ be a superintelligent economy and $(x^*,a^*)$ an
autonomy-Pareto optimum. Then, for the standard supporting-hyperplane
decentralization, $(x^*,a^*)$ is certifiably decentralizable by prices,
lump-sum transfers, and admissible rights assignments---that is, there
exist $(p^*,T^*,\rho^*)$ supporting $(x^*,a^*)$ in the sense of
Definition~3, with the realized profile verifiably matching the one the
support was constructed for---only if all of the following hold:
\begin{enumerate}
\item[(i)] the supporting normal to $U^{\alpha,\mu}$ at the welfare vector
$u^*$ induced by $(x^*,a^*)$ exists, either because $U^{\alpha,\mu}$ is
convex, or because $U^{\alpha,\mu}$ admits a convexification
$\widehat U^{\alpha,\mu}$ whose supporting normal at $u^*$ also supports
$U^{\alpha,\mu}$ at $u^*$;
\item[(ii)] the welfare-bearing support $I^\mu$ and any institutional
welfare weights used beyond Pareto support are fixed at $(x^*,a^*)$, or
are generated by a procedure $\Sigma$ invariant across the institutional
contexts supporting $(x^*,a^*)$;
\item[(iii)] every non-fungible autonomy-right component required by
$a^*$ is assigned, enforced, or protected by $\rho^*$ in a manner
compatible with $(x^*,a^*)$;
\item[(iv)] for every $i \in I_S$ with superposed preferences over
$\Theta_i$ there is a welfare selector $\sigma_i$ that is support-stable
at $(x^*,a^*)$ for the candidate support $(p^*,T^*,\rho^*)$;
\item[(v)] after $T^*$ is realized but before exchange, no agent $j$ can
impose an unpriced support-relevant manipulation externality on the
preference-formation mapping $\Phi_i$ of any welfare-bearing agent $i$. Benign or anticipated preference updating is permitted when it leaves the
supported strict upper contour empty;
\item[(vi)] self-modification and identity-continuity transitions are
either irrelevant to $(x^*,a^*)$, priced in $p^*$, or institutionally
governed by $\rho^*$;
\item[(vii)] the institutional observation map is verification-complete
at $(x^*,a^*,p^*,T^*,\rho^*)$: every realized strategy and
allocation-rights profile observationally equivalent to the target either
implements a welfare-equivalent allocation-rights profile under the
selected welfare relations or is detectable as non-compliance.
\end{enumerate}

\par\smallskip
\end{theorem}

\begin{proof}
Suppose that $(x^*,a^*)$ is certifiably decentralizable by the standard
supporting-hyperplane construction. Thus there exist
$(p^*,T^*,\rho^*)$ satisfying Definition~3, and the support is obtained
from a separating hyperplane between the strict preferred set and the
feasible allocation-rights set in the autonomy-extended space. For each welfare-bearing agent $i\in I^\mu$, let
$P_i^\alpha(x_i^*,a_i^*)$ and $P_i^{\alpha,>}(x_i^*,a_i^*)$ denote the
weakly and strictly preferred sets at $(x_i^*,a_i^*)$. The relevant
strict product preferred set is
\begin{align*}
\mathcal{P}^{\alpha,>}(x^*,a^*) &
:=
\Big\{((x_i,a_i))_{i\in I^\mu} :
(x_i,a_i)\succeq_i^\alpha(x_i^*,a_i^*)\ \forall i\in I^\mu,
\\&\text{ and }(x_j,a_j)\succ_j^\alpha(x_j^*,a_j^*) \text{for some} j
\in I^\mu
\Big\}.
\end{align*}
Let
\[
Z^\mu\subseteq \prod_{i\in I^\mu}(X_i\times A_i)
\]
be the projection of the feasible allocation-rights set onto the
welfare-bearing agents. Since $(x^*,a^*)$ is autonomy-Pareto optimal,
\[
\mathcal{P}^{\alpha,>}(x^*,a^*)\cap Z^\mu=\varnothing .
\]
The standard support construction therefore requires a non-zero
continuous linear functional
\[
L((x_i,a_i)_{i\in I^\mu})
=
\sum_{i\in I^\mu}p_i^*\cdot x_i
+
\sum_{i\in I^\mu}\pi_i^*\cdot a_i
\]
and a scalar $\gamma$ such that $L(z)\ge \gamma$ on
$\mathcal{P}^{\alpha,>}(x^*,a^*)$ and $L(z)\le \gamma$ on $Z^\mu$,
after the usual normalization at the supported allocation. Under the
standard commodity-market assumptions, the commodity components of this
functional reduce to a common price vector $p^*\in\mathbb{R}_+^\ell$,
while the $\pi_i^*$ components are dual valuations of the
agent-indexed autonomy-rights coordinates. We write $q^*$ for the
welfare-space normal, i.e., $L$ expressed in utility coordinates
$(u_i)_{i\in I^\mu}$, when the argument passes to $U^{\alpha,\mu}$. We now show that each listed condition is required for this certification
to be well defined and to support the same allocation-rights profile
through exchange.

\emph{(i).} The separation step requires a supporting normal at the
welfare vector $u^*$ induced by $(x^*,a^*)$: there must exist
$q^*\neq 0$ such that
\[
q^*\cdot u\le q^*\cdot u^*
\qquad\text{for all }u\in U^{\alpha,\mu}.
\]
Convexity of $U^{\alpha,\mu}$, together with the usual closedness and
boundary conditions, is the classical sufficient route to such a normal.
When the construction instead uses a relaxed or convexified set
$\widehat U^{\alpha,\mu}$, the resulting normal certifies the original
economy only if it supports the original set $U^{\alpha,\mu}$ at the same
welfare vector $u^*$. A normal that supports only a relaxed allocation, a
lottery, or a frontier point introduced by the relaxation is not a
certificate for $(x^*,a^*)$. This gives condition~(i); a sufficient
frontier-preservation criterion is stated in
Proposition~\ref{prop:convexification}.

\emph{(ii).} Both the product preferred set
$\mathcal{P}^{\alpha,>}(x^*,a^*)$ and the feasible set $Z^\mu$ are
defined with index running over $I^\mu=\{i\in I:\mu(i)>0\}$. If $\mu$
varies across the institutional contexts relevant to support, the
separation problem itself varies: the sets being separated are
different in each context, and a hyperplane separating them in one
context is in general not a separating hyperplane in another. A single
certificate $L$ therefore cannot exist. Stability of $\mu$, either by
direct fixing at $(x^*,a^*)$ or by an institutionally invariant
generating procedure $\Sigma$, is the condition that collapses the
family of separation problems to one. This gives condition~(ii).

\emph{(iii).} Consider a non-fungible autonomy-right component
$a_i^{*k}$ required by the target allocation. If this component is not
assigned, enforced, or protected by $\rho^*$, then the target pair
$(x_i^*,a_i^*)$ is not attainable in agent $i$'s rights-feasible budget
set. Prices and lump-sum transfers can relax only the commodity side of
the budget constraint. By non-fungibility, no commodity compensation
vector $\tau$ can make a bundle with the $k$-th right altered
welfare-equivalent to $(x_i^*,a_i^*)$. Thus the dual term $\pi_i^{*k}$
cannot be projected into an ordinary scalar wealth transfer. The
support fails unless the right is assigned, enforced, or protected by
$\rho^*$ in a way compatible with $(x^*,a^*)$. This gives
condition~(iii).

\emph{(iv).} Suppose some $i\in I_S$ has preferences superposed over
$\Theta_i$. The supporting hyperplane is constructed from the strict
upper contour of $(x_i^*,a_i^*)$ on the supported budget set. If no
support-stable selector is available, then different institutional
contexts may generate different intersections
\[
P_i^{\sigma_i(c),>}(x_i^*,a_i^*)\cap B_i(p^*,T^*,\rho^*),
\]
so a bundle that is excluded by the support in one context may be
strictly preferred and budget-feasible in another. The support therefore
does not certify the same choice problem across the relevant contexts.
Support-stability of the welfare selector is required. This gives
condition~(iv).

\emph{(v).} The support must remain valid between the realization of
transfers and the exchange stage. Let $t_0$ be the announcement of
$(p^*,T^*,\rho^*)$, $t_1$ the realization of $T^*$, and $t_2$ the
exchange stage. The separating hyperplane is derived at $t_0$ from the
pre-intervention preference-formation mapping $\Phi_i$. If, after
$T^*$ is realized but before exchange, some agent $j$ can choose an
action $y_j$ that induces a support-relevant unpriced manipulation
$m_{ij}(y_j)(\Phi_i)\neq\Phi_i$, then the relation governing $i$'s
choice at $t_2$ is $\succeq_i^{\alpha,m}$ rather than the relation used
to construct the support. If
\[
P_i^{\alpha,m,>}(x_i^*,a_i^*)\cap B_i(p^*,T^*,\rho^*)\neq\varnothing,
\]
then $(x_i^*,a_i^*)$ is not maximal on the realized budget set and the
support fails at $i$. Benign or anticipated updating that leaves this
intersection empty does not create this failure. This gives
condition~(v).

\emph{(vi).} Self-modification by $i$ between $t_0$ and $t_2$ alters
$\Phi_i$, and hence the relation $\succeq_i^\alpha$ on which
$P_i^\alpha(x_i^*,a_i^*)$ was constructed; identity transitions (copy,
merge, split) alter the index set $I^\mu$ after the welfare-bearing
population was fixed at $t_0$. The hyperplane $L$ was constructed
against $(\succeq_i^\alpha)_{i\in I^\mu}$ with $I^\mu$ fixed. Unless
the transition is irrelevant to $(x^*,a^*)$ (so the pre- and
post-transition welfare objects coincide at $(x^*,a^*)$), priced in
$p^*$ (so the resource cost of transitioning enters the commodity
budget), or governed by $\rho^*$ (so the transition is a permitted move
within the rights profile), $L$ fails to separate the post-transition
analogues of $\mathcal{P}^{\alpha,>}(x^*,a^*)$ and $Z^\mu$. This gives
condition~(vi).

\emph{(vii).} Let $\mathcal{O}$ denote the
institution's observation map from realized strategies and
allocation-rights profiles into verifiable market data. Certifiability
requires that whenever a realization $s$ is observationally equivalent
to the target realization $s^*$---that is,
$\mathcal{O}(s)=\mathcal{O}(s^*)$---either $s$ implements an
allocation-rights profile welfare-equivalent to $(x^*,a^*)$ under the
selected welfare relations, or the deviation is detectable as
non-compliance. If this condition fails, there exists an
observationally equivalent realization that passes the institution's
market-data checks while implementing a different welfare object. In
that case market clearance and budget maximization do not certify
compliance with $(p^*,T^*,\rho^*)$. This gives condition~(vii).
\end{proof}

Conditions (i)--(vii) are necessary for the standard
supporting-hyperplane construction to certify $(p^*,T^*,\rho^*)$ (the
theorem does not rule out non-hyperplane supports). As a \emph{specialization} (not a
converse), when preferences are convex, autonomy rights reduce to ordinary
commodities or fixed rights, $\mu$ is settled, each relevant preference
family is singleton, $\{m_{ij}\}=\emptyset$, self-modification is absent,
and verification is complete, the classical SWT
\cite{debreu1959theory,mwg1995microeconomic} is recovered. The following propositions make these mechanisms explicit by isolating
the corresponding failure modes.

\begin{proposition}[Non-Decentralizability under Non-Fungible Autonomy Rights]\label{prop:nonfungible}
If an autonomy-Pareto optimum $(x^*,a^*)$ depends on a non-fungible
autonomy-right component $a_i^{*k}$ for some welfare-bearing $i$, and no
admissible $\rho\in R$ assigns, enforces, or protects $a_i^{*k}$ for $i$
at $(x^*,a^*)$, then prices and lump-sum transfers alone cannot
decentralize $(x^*,a^*)$.
\end{proposition}

\begin{proof}
Since no admissible $\rho\in R$ assigns, enforces, or protects
$a_i^{*k}$ for $i$, the target pair $(x_i^*,a_i^*)$ is not contained in
the rights-feasible part of $i$'s budget set. Prices and scalar
transfers alter the commodity budget but cannot add the missing right.
By non-fungibility, no commodity compensation vector yields a
welfare-equivalent substitute with the $k$-th right altered and the
other autonomy coordinates held fixed. Hence under prices and transfers
alone, $(x_i^*,a_i^*)$ is not budget-feasible, let alone
$\succeq_i^\alpha$-maximal.
\end{proof}
Proposition~\ref{prop:nonfungible} shows that when autonomy rights are
non-fungible, no price--transfer scheme can reproduce the required
rights allocation, so decentralization fails even when a supporting
hyperplane exists.

\begin{proposition}[Manipulation Failure]\label{prop:manipulation}
Let $(p^*,T^*,\rho^*)$ support $(x^*,a^*)$ via the hyperplane of
\S\ref{subsec:hyperplane}, constructed for the pre-manipulation relation
$\succeq_i^\alpha$. Suppose there exist a welfare-bearing agent $i\in I^\mu$, an agent
$j\in I$, and an action $y_j\in Y_j$ feasible after $T^*$ but before
exchange such that
$m_{ij}(y_j)(\Phi_i)\neq\Phi_i$ and the induced post-manipulation relation
$\succeq_i^{\alpha,m}$ strictly alters the preferred set at
$(x_i^*,a_i^*)$, in the sense that
$P_i^{\alpha,m,>}(x_i^*,a_i^*)\cap B_i(p^*,T^*,\rho^*)\neq\varnothing$.
Then $(x^*,a^*)$ is not decentralized by $(p^*,T^*,\rho^*)$.
\end{proposition}

\begin{proof}
Because the manipulation $m_{ij}(y_j)$ acts between $t_1$ and $t_2$,
the preference-formation mapping $\Phi_i$ used to construct the
supporting hyperplane at $t_0$ is replaced by a new mapping before
exchange, yielding the post-manipulation relation
$\succeq_i^{\alpha,m}\neq\succeq_i^\alpha$ on $X_i\times A_i$. The
realized choice problem at $t_2$ is therefore evaluated under
$\succeq_i^{\alpha,m}$ on $B_i(p^*,T^*,\rho^*)$. The hypothesis
$P_i^{\alpha,m,>}(x_i^*,a_i^*)\cap B_i(p^*,T^*,\rho^*)\neq\varnothing$
exhibits a feasible bundle strictly preferred to $(x_i^*,a_i^*)$ under
$\succeq_i^{\alpha,m}$, so $(x_i^*,a_i^*)$ is not
$\succeq_i^{\alpha,m}$-maximal on its budget set and the support
fails at $i$.
\end{proof}
Together, these propositions identify two distinct ways in which the
supporting-hyperplane construction can fail: through the inability to
reduce autonomy rights to wealth transfers, and through instability of
the preference relation over the decentralization timeline. We now turn
to the remaining condition governing the existence of the supporting
hyperplane itself.

\subsection{Frontier-preserving convexification}
\label{subsec:convexification}
Even in the absence of these failures, the supporting hyperplane may not
exist in the autonomy-extended setting if the welfare possibility set is
nonconvex. The following condition ensures that convexification does not
introduce spurious support at the relevant welfare point. Condition~(i) rules out the failure mode identified in the literature  \cite{starrett1972nonconvexities}: once rights are bundled with
allocations, $U^{\alpha,\mu}$ can have ``holes'' or exposed
nonconvexities at Pareto-relevant points, so a supporting hyperplane can
exist for $\widehat U^{\alpha,\mu}=\mathrm{co}(U^{\alpha,\mu})$ without
any genuine support for $U^{\alpha,\mu}$ itself. Condition~(i) requires
that convexification not create a new frontier face at $u^*$.

\begin{proposition}[Purification as frontier-preserving convexification]\label{prop:convexification}
Consider the continuum-agent extension of the welfare possibility set.
Let $(S,\mathcal S,\lambda)$ be a nonatomic probability space of agent
characteristics, let $n:=|I^\mu|<\infty$, and let
$G:S\rightrightarrows \mathbb{R}^n$ be a measurable correspondence with
nonempty compact values, integrably bounded by an integrable function
$b:S\to\mathbb{R}_+$:
\[
\|g\|\le b(s)\qquad\text{for all }g\in G(s)\text{ and a.e. }s.
\]
Interpret $G(s)$ as the set of welfare-contribution vectors attainable by
type $s$ from feasible commodity-rights choices. Define the aggregate
welfare possibility set and its pointwise convexification by
\[
U_\lambda
:=
\left\{
\int_S g(s)\,d\lambda(s)-r:
g(s)\in G(s)\ \text{a.e.},\ r\in\mathbb{R}^n_+
\right\},
\]
and
\[
\widehat U_\lambda
:=
\left\{
\int_S h(s)\,d\lambda(s)-r:
h(s)\in \operatorname{co}G(s)\ \text{a.e.},\ r\in\mathbb{R}^n_+
\right\}.
\]
Then
\[
U_\lambda=\widehat U_\lambda.
\]
Consequently, aggregate convexity of the welfare possibility set follows
from nonatomicity even when the individual correspondence $G(s)$ is
nonconvex. In particular, if $u^*\in U_\lambda$ is supported by a normal
$q^*$ in $\widehat U_\lambda$, then the same $q^*$ supports
$U_\lambda$ at $u^*$.
\end{proposition}

\begin{proof}
Because $G$ is measurable, compact-valued, and integrably bounded, the
Aumann integral $\int_S G\,d\lambda$ is well defined. Since
$G(s)\subseteq \operatorname{co}G(s)$ for a.e. $s$, we have
$U_\lambda\subseteq \widehat U_\lambda$. For the reverse inclusion, take any
\[
\widehat u=\int_S h(s)\,d\lambda(s)-r\in \widehat U_\lambda,
\qquad h(s)\in \operatorname{co}G(s)\ \text{a.e.}
\]
The nonatomicity of $\lambda$ and the purification theorem for Aumann
integrals imply that every integrable selection of the convexified
correspondence $\operatorname{co}G$ has the same integral as some
integrable selection $g$ of the original correspondence $G$:
\[
\int_S h(s)\,d\lambda(s)=\int_S g(s)\,d\lambda(s),
\qquad g(s)\in G(s)\ \text{a.e.}
\]
Therefore $\widehat u=\int_S g(s)\,d\lambda(s)-r\in U_\lambda$, so
$\widehat U_\lambda\subseteq U_\lambda$. Hence
$U_\lambda=\widehat U_\lambda$. The final claim follows immediately: if $q^*$ supports
$\widehat U_\lambda$ at $u^*$, then
$q^*\cdot u\le q^*\cdot u^*$ for every $u\in \widehat U_\lambda$; since
$U_\lambda=\widehat U_\lambda$, the same inequality holds for every
$u\in U_\lambda$.
\end{proof}
This proposition gives a precise sense in which condition~(i) can be
weakened: individual nonconvexities need not defeat the supporting-normal
step when the relevant heterogeneity is represented by a nonatomic
population and the welfare possibility set is an Aumann aggregate. Other frontier-preserving convexifications are not ruled out, but they must be established directly.  

\subsection{Manipulation timeline}
\label{subsec:timeline}

Write $t_0$ for the announcement of $(p^*,T^*,\rho^*)$, $t_1$ for the
realization of transfers, and $t_2$ for exchange.

\begin{lemma}[Preference drift between transfer realization and exchange]\label{lem:timeline}
Suppose $(p^*,T^*,\rho^*)$ supports $(x^*,a^*)$ via the hyperplane of
\S\ref{subsec:hyperplane}. If agent $j$ acts between $t_1$ and $t_2$ via
some $y_j\in Y_j$ that induces a support-relevant manipulation
externality on a welfare-bearing $i\in I^\mu$, then the choice problem
actually faced by $i$ at $t_2$ is evaluated against a preference relation
$\succeq_i^{\alpha,m}$ different from the relation used to derive the
supporting hyperplane at $t_0$.
\end{lemma}

\begin{proof}
The separating argument fixes $P_i^\alpha(x_i^*,a_i^*)$ from the relation
induced by $\Phi_i$ at $t_0$. By the definition of a support-relevant
manipulation externality, $m_{ij}(y_j)(\Phi_i)\neq\Phi_i$ yields a new
preference-formation mapping, hence a new relation
$\succeq_i^{\alpha,m}$ on $X_i\times A_i$. The commodity budget may be
unchanged, but the maximization problem at $t_2$ is now a
$\succeq_i^{\alpha,m}$-maximization problem. If the manipulated strict
upper contour intersects $B_i(p^*,T^*,\rho^*)$, the original
$t_0$-hyperplane does not certify maximality of $(x_i^*,a_i^*)$.
\end{proof}

\section{Discussion and Conclusion}
\label{sec:discussion}
Our aim has been to identify which parts of classical welfare theory survive into superintelligent economies unchanged, which require conditional
restatement, and which require institutional machinery beyond the
price-transfer apparatus. The broader implication is that the SWT's classical
separation of equity from efficiency is conditional on features of the
welfare object that classical theory could take for granted and
superintelligent economics may not.  The autonomy-qualified theorem reframes decentralization in
superintelligent economies as a joint problem over three institutional
objects rather than one: a price vector $p^*$, a transfer profile
$T^*$, and a rights assignment $\rho^*$, with a welfare-selector
mechanism $\sigma$ sitting behind condition~(iv). $\rho^*$ and $\sigma$ must be institutionally specified, which means decentralization cannot be read off from
equilibrium prices alone. Conditions (ii), (v), (vi) and (vii) further
demand ex ante procedural machinery---status rules, anti-manipulation
windows, self-modification governance, oversight
infrastructure---before any support step can be attempted. The
classical separation between equity and efficiency therefore survives
only in the regime where conditions (ii)--(vii) hold; outside that
regime, redistribution and efficiency are coupled through $\rho^*$ and
$\sigma$. Finally, some clarifications are in order. The theorem retains force even if economic preference superposition is read as metaphorical, since condition~(iv) is triggered by the welfare-economic definition of \S\ref{sec:formalism} and the alignment literature already shows a single rationalizing relation is not always available~\cite{hadfieldmenell2017inverse,casper2023open,gabriel2020values,conitzer2024social}; an expected-utility reply over $\Theta_i$ does not escape this, since selector stability is exactly the requirement that the implicit prior be context-invariant. Future research may seek to explore these trends, including the disruptive role that superintelligent systems pose to assumptions of autonomy \cite{perrier2026post} underpinning welfare economics.

\bibliographystyle{splncs04}
\bibliography{refs}

\end{document}